\documentclass{article}

\usepackage[utf8]{inputenc} 

\usepackage{amssymb, amsthm, amsmath} 

\newcommand{\F}{\mathbb F}
\newcommand{\Fqd}{{\mathbb F}_{q^2}}

\newtheorem{definition}{Definition}[section]
\newtheorem{proposition}[definition]{Proposition}

\newtheorem{lemma}[definition]{Lemma}
\newtheorem{corollary}[definition]{Corollary}

\usepackage[a-1b]{pdfx}
\usepackage{hyperref}
\begin{document}

\title{Improving the dimension bound of Hermitian Lifted Codes}
\author{ Austin Allen \and Eric Pab\'on--Cancel \and Fernando Pi\~nero--Gonz\'alez \and Lesley Polanco}

\maketitle
\begin{abstract} In this article we improve the dimension and minimum distance bound of the  the Hermitian Lifted Codes LRCs construction from  L\'opez, Malmskog, Matthews, Pi\~nero and Wooters (L\'opez et. al.) via elementary univariarte polynomial division. They gave an asymptotic rate estimate of $0.007$. N. Nevo genealized the rate for general $p$. Foe example the asymptotic rate for Hermitian Lifted Codes is $0.000152$ in the ternary case, $p = 3$.  For the case where $q$ is a power of $2$ we improve the rate estimate to $0.010$ using univariate polynomial division.\end{abstract}
\section{Introduction}

A locally recoverable code (LRC) is a linear code that can recover a single erased position from a small set of coordinates. Tamo and Barg developed optimal LRCs from subcodes of Reed-Solomon codes. Guo et al. \cite{GKS13}  employed the point-line geometry of affine spaces over $\mathbb{F}_q$ to construct LRCs. Subsequently, L\'opez et al. \cite{LMMPW21} used the point-line incidence of an affine part of the Hermitian curve to define LRCs. In \cite{LMMPW21} an asymptotic rate bound of $0.007$ on Hermitian Lifted codes was established for $p =2$. In subsequent work \cite{N22} N. Nevo generalized the bound to arbitrary primes $p$. The generalized rate is $\frac{0.469}{p^4(p-1)(p^3-p^2-1)}$. This rate bound decreases with $p$, but the asymptotic rate bound stays positive for fixed $p$. Now we present some fundamental concepts of codes with locality.

\subsection{Locality and Availability}
    \begin{definition}[Locality of a Linear Code]\cite{GHSY12}
    A code $C$ has \emph{locality} $r$ if for every $i \in [n]$ there exists a subset $R_{i} \subset [n] \setminus i,\| R_{i} \| \leq r$ and a function $\phi_{i}$ such that for every codeword  $c \in C$:
    \begin{equation}
        c_{i} = \phi_{i, R_i} (\{c_{j}, j \in R_{i}\})
    \end{equation} where the recovery function $ \phi_{i, R_i}$ depends on the position $i$ and the recovery set $R_i$ used. 
    \end{definition}
    \begin{definition}[Availability of a code with locality] A code $C$ with locality $r$ has \emph{availability} $s$ if for any $i \in [n]$ there exists $s$ disjoint subsets $R_{i,1}, R_{i,2}, \ldots, R_{i,s}$ of size at most $r$ which may be used to recover $c_i$.
    \end{definition}
A linear code with locality $r$ and availability $s$ is a linear code where any position $i$ can be recovered from any of $s$ disjoint sets, each of size at most $r$.
\subsection{Hermitian codes as evaluation codes}
\begin{definition}[Affine Points of the Hermitian curve]
Let $q$ be a prime power. The \emph{affine points of the Hermitian curve} over $\Fqd$ are the solutions to $$X^{q+1} = Y^q+Y$$ over $\Fqd$. That is, the points are defined by $$ \mathcal{H} := \{ (\alpha, \beta) \in \Fqd^2 \ \vert \ \alpha^{q+1} = \beta^q+\beta  \} $$

\end{definition}

Hermitian codes may be defined as evaluation codes of polynomials over $\mathcal{H}$. Since $\mathcal{H}$ is finite any function on $\mathcal{H}$ may be described as a linear combination of a finite set of monomials. One such set is given as follows.

\begin{definition}
Denote by $\mathcal{M}$ the vector space spanned by the following monomials $$\mathcal{M}  := \langle  X^iY^j \ \vert 0 \leq i < q^2, 0 \leq j < q\rangle_{\F_{q^2}}.$$
\end{definition}
We also define the evaluation of a polynomial on a set.
\begin{definition}

Let $f \in \F_{q^2}[X,Y]$. Let $V = \{P_1, P_2, \ldots, P_n\} \subseteq \F_{q^2}^2$. We denote the evaluation of $f$ on $V$ by $$ev_V(f) = (f(P_1),f(P_2), \ldots, f(P_n)) $$
\end{definition} The ideal of functions vanishing on $\mathcal{H}$, the ideal $$I(\mathcal{H)} = \langle X^{q+1}-Y^q-Y, X^{q^2}-X, Y^{q^2}-Y \rangle$$ equals $$I(\mathcal{H}) = \langle X^{q+1}-Y^q-Y, X^{q^2}-X \rangle.$$ With the theory of Gr''obner bases the following propositions can be established. The readers interested may consult \cite{CLO07}
\begin{proposition}
Let $q$ be a prime power. Let $\mathcal{H}$ denote the set of all points over $\F_{q^2}$ of the Hermitian curve. Let $f \in \F_{q^2}[X,Y]$ be any polynomial. Then there exists $g \in \mathcal{M}$ such that $$ev_\mathcal{H}(f) = ev_\mathcal{H}(g).$$
\end{proposition}
\begin{proposition}
Let $q$ be a prime power. Let $\mathcal{H}$ denote the set of all points over $\F_{q^2}$ of the Hermitian curve. Let $f, g\in \mathcal{M}$.  Then $$f = g \makebox{ if and only if } ev_\mathcal{H}(f) = ev_\mathcal{H}(g).$$
\end{proposition}


Evaluation codes are defined as linear codes obtained by evaluating a certain set of polynomials over a set of finite points. We define evaluation codes defined over $\mathcal{H}$.

\begin{definition}[Evaluation codes over the Hermitian curve]\cite{G08}
Let $L$ be an $\F_{q^2}$--linear subspace of  $\mathcal{M}$. An evaluation code over $\mathcal{H}$ is defined as:

$$C(L, \mathcal{H}) := \{ ev_{\mathcal{H}}(f) \ \vert \ f \in \langle L  \} $$
\end{definition}

Algebraic function fields establish bounds on length, dimension and minimum distance of Hermitian codes. In contrast, we define Hermitian codes evaluating an explicit set of (monomial) functions on a explicit set of points. The Hermitian code may be defined as $C(\mathcal{M}(s), \mathcal{H})$ where  $$\mathcal{M}(s) := \{X^iY^j \in \mathcal{M} \ \vert \ qi+(q+1)j \leq s \}.$$ Full details on the definition of Hermitian codes as evaluation codes may be found in \cite{JP14}.

When Hermitian codes are defined using evaluation codes, Gröbner bases can be employed to calculate their dimension and minimum distance. The ideal $I(\mathcal{H}) = \langle X^{q+1}-Y^q-Y, X^{q^2}-X, Y^{q^2}-Y \rangle$ is the kernel of the evaluation map over $\mathcal{H}$  map for those points. This implies it may be easy to determine the dimension of any evaluation code. With an explicit basis of independent functions, certain computations can be simplified, and the footprint bound can be utilized to obtain lower bounds on the minimum distance.

\subsection{Lines of the Hermitian curve}

We utilize the geometry of the Hermitian curve to construct a locally recoverable code (LRC). This approach is similar to the one employed by Guo, Kopparty, and Sudan,\cite{GKS13} who construct Reed-Solomon lifted codes using lines of affine spaces. The locality condition requires that any polynomial function reduces to a function of degree $\leq q-2$ when restricted on any line. However, for Reed-Solomon lifted codes, the low degree condition may achieved utilizing the $(0,1)$--characteristic vectors of each line as parity check equations. This implies Lifted Reed--Solomon codes have very good rate and the LRC can be considered over the prime field $\mathbb{F}_p$.

Remarkably, the linear code associated with the lines of the Hermitian unital has a dimension of $q^3+1$ over $\mathbb{F}_p$, implying that any code utilizing the characteristic vector of each line of the Hermitian unital has a dimension of $0$. Our LRCs are linear codes which employ parity check equations with the same nonzero positions as the linear code associated to the Hermitian unital but with a high dimension.

\begin{definition}[Lines of the Hermitian curve]\cite{LMMPW21}
Let $q$ be a prime power. Let $a,b \in \Fqd$. A \emph{line of the Hermitian curve} is a set $L_{a,b}$ of the form $$L_{a,b} := \{(x,y) \in \mathcal{H} \ \vert y = ax+b \} \makebox{ and } \#(L_{a,b}) = q+1.$$

\end{definition} 

The Hermitian unital is a collection of $q^3+1$ points in $\mathbb{P}^2(\F_{q^2})$ isotropic under a nondegenerate Hermtian form. All lines of the projective plane intersect the Hermitian unital in either $1$ or $q+1$ places. We are interested in an affine map of the Hermitian unital, which contains $q^3$ points only. To recover positions in our code, we use the pointsets of lines of the Hermitian unital which intersect the affine part on $q+1$ points. Our selected functions are those with degree  $\leq q-1$ when restricted to any such line. The $x$-coordinates of the points on the lines of the Hermitian curve satisfy a particular polynomial equation of degree $q+1$. 

\begin{definition}
Let $q$ be a prime power. Let $a,b \in \Fqd$. Define by $L_{a,b,x}$ the set of $x$--coordinates of the line $L_{a,b}$. That is: $$L_{a,b,x} = \{ x \ \vert \ (x,y) \in L_{a,b} \}.$$

\end{definition}

\begin{lemma}\label{lem:linecond}

Let $a,b \in \Fqd$. Then the points in $L_{a,b,x}$ satisfy the univariate polynomial equation:   $$(X-a^q)^{q+1} -(a^{q+1}+b^q+b).$$

\end{lemma} \begin{proof} We need to determine the common points to $Y = aX+b$ and $X^{q+1} = Y^q +Y$ over $\Fqd$. Substitute $Y = aX+b$ on the equation of the Hermitian curve to obtain:
 
 $$X^{q+1} = (aX+b)^q+(aX+b).$$ We rearrange terms and obtain:
  $$X^{q+1}-a^qX^q -aX = b^q+b.$$ We add $a^{q+1}$ to both sides.
  $$X^{q+1}-a^qX^q -aX +a^{q+1} = b^q+b + a^{q+1}.$$ The right hand side factors as: $$(X^q-a)(X-a^q) = b^q+b + a^{q+1}.$$
  Because $a \in \F_{q^2}$ note that $X^q-a = (X-a^q)^q$.
  Therefore $$(X^q-a)(X-a^q)  = (X-a^q)^q(X-a^q)=(X-a^q)^{q+1} = b^q+b + a^{q+1}.$$ Thus the elements of $L_{a,b,x}$ satisfy  $$(X-a^q)^{q+1} -(a^{q+1}+b^q+b) = 0.$$ \end{proof}

Now we state the condition on $a,b$ such that $L_{a,b}$ is a line of the Hermitian curve.
\begin{lemma}\cite{LMMPW21} Let $a,b \in \Fqd$ be such that $L_{a,b}$ is a line of the Hermitian curve. Then $a^{q+1} + b^q+b \neq 0$.

\end{lemma}
 \begin{proof} Let $L_{a,b}$ be a line of the Hermitian curve. The $x$--coordinates satisfy the polynomial equation  $(X-a^q)^{q+1} -(a^{q+1}+b^q+b) = 0$. Note that since $b^q+b + a^{q+1} \in \F_q$, the equation has $1$ solution if $b^q+b + a^{q+1} = 0$ and $q+1$ solutions if $b^q+b + a^{q+1} \neq 0$. For each solution in $X$ there is one point in $L_{a,b}$. Therefore if $L_{a,b}$ is a line of the Hermitian curve, it has $q+1$ points. This implies $(X-a^q)^{q+1} = b^q+b + a^{q+1}$ has $q+1$ solutions and therefore $b^q+b + a^{q+1} \neq 0$. \end{proof}
 
 The nonzero positions of the parity check equations for the Hermitian lifted code correspond to the point sets of the lines in the Hermitian unital. The linear code generated by the $(0,1)$ characteristic vector of those lines has a dimension of $q^3+1$ \cite[Theorem 8.3.1]{AK92}). Consequently, the LRCs defined by the lines of the Hermitian unital have a dimension of $0$. It is noteworthy that, despite sharing the same nonzero positions for the parity check equations, the Hermitian lifted codes exhibit a relatively large dimension.

\begin{definition}\cite{LMMPW21} Let $f(X,Y)$ be a bivariate polynomial. Let $L_{a,b}$ be a line of the Hermtian curve. The \emph{restriction of $f$ onto $L_{a,b}$} is the function obtained by evaluating $f$ on the points of the line $L$. We denote the restriction by $f_{L_{a,b}}$. \end{definition}

It is important to differentiate between a polynomial and its evaluation. If the line $L_{a,b}$ is represented by the equations $X = T$ and $Y = aX+b$, then $f_{L_{a,b} } = f(T, aT+b)$, which is a univariate polynomial on $T$. The restriction of the evaluation $ev_{\mathcal{H}}(f)$ to the line $L_{a,b}$ is simply $ev_{L_{a,b,x}}f_{L_{a,b}}(T)$.

Even if $f(X,Y)$ has a high degree, the evaluation vector $ev_{L_{a,b}}(f)$ may correspond to the evaluation of polynomial of degree $q-1$ or less. If this degree condition is held for all lines then one can make a locally recoverable code (LRC). However, to achieve this, we require functions that restrict in a desirable manner on each line.

\begin{definition}[Good functions]\cite{LMMPW21}
Let $f(X,Y)$ be a bivariate polynomial. Let $a,b \in \Fqd$ such that $L_{a,b}$ is a line of the Hermitian curve. We say $f$ is a \emph{good polynomial} if and only the evaluation $$ev_{L_{a,b}}(f) = ev_{L_{a,b,x}}(g(T))$$ where $g$ is a univariate polynomial of degree  less than $q$ on $L_{a,b,x}$ for each line $L_{a,b}$ of the Hermitian curve.

We denote the set of good functions as $\mathcal{G}_f$ and the set of good monomials as $$\mathcal{G}_M := \{X^iY^j \in \mathcal{M} \vert  X^iY^i \makebox{ is a good function }  \}.$$
\end{definition}

\begin{definition}[Hermitian Lifted Codes]\cite{LMMPW21}

We define the \emph{Hermitian Lifted code} over the Hermitian curve as $$\mathcal{C}  := C(\mathcal{G}_f, \mathcal{H}).$$ The Hermitian LRC of good monomials is defined as $$\mathcal{C}_M := C(\mathcal{G}_M, \mathcal{H}).$$

\end{definition}

Note that $\mathcal{C}_M$ is  a subcode of $\mathcal{C}_f$.. L\'opez et. al \cite{LMMPW21} claim the following.

\begin{proposition} \cite[Claim 12]{LMMPW21}
Let $q = 2^k$. Then the monomial set $\mathcal{M}$ contains at least $$\sum\limits_{r=0}^{k-1}(4^r-3^r)4^{k-r-2}2^{k-r-1}$$ good monomials.  Consequently $$\dim \mathcal{C} \geq \dim \mathcal{C}_M \geq \sum\limits_{r=0}^{k-1}(4^r-3^r)4^{k-r-2}2^{k-r-1}.$$
\end{proposition}

Nevo's previous work on the $p$--ary case \cite{N22} also determine an asymptotic rate bound which depends only on $p$ and not on $q$.

\begin{proposition} \cite[Theorem 5]{N22}
Let $q = p^k$. Then $$\dim \mathcal{C} \geq \dim \mathcal{C}_M \geq \frac{0.469}{p^4(p-1)(p^3-p^2-1)}.$$
\end{proposition}

As a corollary, it has been demonstrated that the rate of Hermitian Lifted codes satisfies the lower bound $R \geq 0.007$. While Hermitian-Lifted codes can be defined over any characteristic, both our dimension analysis and the analysis presented in \cite{LMMPW21} were conducted specifically for even characteristic. This choice was made to streamline computations and facilitate analysis.

\section{The degree of $T^j \mod (T-a^q)^{q+1}-\gamma$}

Our objective is to discover additional monomials in $\mathcal{M}$ that exhibit favorable degree constrains on each line. Our technique is based in univariate polynomial division.  To streamline our reasoning, we introduce the following notation.

\begin{definition}
Let $\gamma \neq 0 $. We shall denote by $$P_{a,\gamma} := (T-a)^{q+1} - \gamma.$$
\end{definition}

Let us recall that for given $a, b \in \mathbb{F}_{q^2}$, the set $L_{a,b}$ comprises all points of the Hermitian curve $X^{q+1} = Y^q+Y$ that also satisfy $Y = aX+b$. The $X$-coordinates of the points in $L_{a^q,b}$ satisfy the univariate polynomial equation $P_{a,\gamma} = 0$, where $\gamma = a^{q+1}+b^q+b$. If $\gamma \neq 0$, then there exist $q+1$ distinct solutions to $P_{a, \gamma} = 0$, indicating that $L_{a^q,b}$ is a line of the Hermitian curve.

Given a function $f(X,Y)$, its restriction on the line $Y = a^qX+b$ can be obtained by the change of variables $X = T$ and $Y = a^qX+b$. Therefore, $f_{L_{a^q,b}} = f(T, a^qT+b)$. Since the $x$-coordinates of the points on $L_{a^q,b}$ satisfy $P_{a,\gamma} = 0$, the function $f(T, a^qT+b)$ is evaluated only on the $q+1$ roots of $P_{a, \gamma}$. Consequently, the goodness or badness of $f(X,Y)$ depends only on $\deg \left( f(T, a^qT+b) \mod P_{a,\gamma} \right)$. In this section, we establish crucial properties of the reduction $T^i \mod P_{a, \gamma}$. Hermitian Lifted Codes are defined for all $q$; however, the case over characteristic $2$ is simpler and can be proven to have more monomials. We begin with the following proposition about the binomial coefficients $\mod p$. 
\begin{proposition}[Lucas' Theorem]\cite{L78}
Let $i,j$ be nonnegative integers such that $i=\sum\limits_{s=0}^m i_sp^s$ and  $j= \sum\limits_{s=0}^mj_sp^s$ where $0 \leq i_s, j_s < p$ for all $0 \leq s \leq m$. Then

$$\binom{j}{i} \equiv \prod_{s=0}^m \binom{j_s}{i_s} \mod p.$$ 
\end{proposition}

\begin{definition}\cite{GKS13}

Let $0 \leq i, j \leq p^{m+1}-1$. Suppose that the expansion of $i$ in base $p$ is $i=\sum\limits_{s=0}^m i_sp^s$ and   the expansion of $j$ in base $p$ is $j= \sum\limits_{s=0}^mj_sp^s$  where $0 \leq i_u,j_u \leq  p-1$. We say $i$ \emph{lies in the $p$--shadow of $j$} if and only if $i_s \leq j_s$ for $0 \leq s \leq m$, We denote the relation by $i \leq_p j$.
\end{definition}

As a corollary, we obtain the following:
\begin{corollary}
Let $i,j$ be nonnegative integers then
$$\binom{j}{i} \equiv 0 \mod p \makebox{ if and only if}  \  i \not\leq_p j$$ 
\end{corollary}

Since the change of variables $T = S+a$ does not change the degree of any polynomial, we shall determine $\deg\left( (S+a)^i \mod S^{q+1}-\gamma \right)$ instead of $\deg \left( T^i \mod P_{a,\gamma}\right)$.

\begin{lemma}\label{lem:reduction Xqpr}
Let $q$ be a power of $p$. Let $j$ be a positive integer relatively prime to $p$ and  let $p^r$ be a power of $p$ such that  $jp^r<q$. Then
$$\deg\left( {T^{jp^rq}\mod P_{a,\gamma}}\right) = q+1-p^r$$

\end{lemma}
\begin{proof}
Under the change of variables $S = T-a$ we shall determine $$\deg\left((S+a)^{jp^rq}\mod S^{q+1}-\gamma\right).$$

Note that $$(S+a)^{jp^rq} = (S^{p^rq}+a^{p^rq})^j = \sum\limits_{j_0=0}^j \binom{j}{j_0}S^{j_0p^rq}a^{(j-j_0)p^rq}.$$
If $j_0 > 0$ then $j_0p^rq = (j_0p^r-1)(q+1) + (q+1-j_0p^r)$. The bounds on $j$ and $p^r$ imply $0 < q+1-j_0p^r \leq q$.  In this case  $$(S+a)^{jp^rq} = a^{jp^rq}+ \sum\limits_{j_0=1}^j \binom{j}{j_0}S^{(jp^r-1)(q+1)+(q+1-j_0p^r)}a^{(j-j_0)p^rq}.$$  Reducing modulo $S^{q+1}-\gamma$ we obtain $$(S+a)^{jp^rq} \equiv a^{jp^rq}+ \sum\limits_{j_0=1}^j \binom{j}{j_0}S^{(q+1-j_0p^r)}a^{(j-j_0)p^rq}\gamma^{j_0p^r-1}.$$  The highest possible degree is attained for $j_0 = 1$. In this case note that $\binom{j}{1} = j \not\equiv 0 \mod p$. Therefore, $$\deg\left(a^{jp^rq}+ \sum\limits_{j_0=1}^j \binom{j}{j_0}S^{(q+1-j_0p^r)}a^{(j-j_0)p^rq}\gamma^{j_0p^r-1}\right) = q+1-p^r $$ which finishes the proof.\end{proof}
%
%
%
%

Now we extend the proof when we multiply by certain powers of $T^{p^r}$.

\begin{lemma}\label{lem:reduction Xqpr2}
Let $q$ be a power of $p$. Let $j$ be a positive integer relatively prime to $p$ and  let $p^r$ be a power of $p$ such that  $jp^r<q$. Let $k$ be a positive integer such that $kp^r <q$ Then
$$\deg\left(T^{jp^rq+kp^r}\mod P_{a,\gamma}\right) \leq q+1-p^r$$

\end{lemma}

\begin{proof}
From Lemma \ref{lem:reduction Xqpr} we know that $$\deg\left(T^{jp^rq} \mod P_{a,\gamma}\right) = q+1-p^r.$$ Now we shall reduce $T^{jp^rq+kp^r} \mod P_{a, \gamma}$ instead. With the change of variables $S = T-a$ we obtain:

$$\left(S+a\right)^{jp^rq+kp^r} = \left(S+a\right)^{jp^rq}\left(S+a\right)^{kp^r} =\left(S^{p^rq}+a^{p^rq}\right)^{j}(S^{p^r}+a^{p^r})^{k} .$$

Note that $kp^r<q$. Applying the binomial theorem to $(S+a)^{kp^r}$ we can write $$(S+a)^{kp^r} = \left(\sum\limits_{k_0=0}^{k} \binom{k}{k_0} S^{k_0p^r}a^{(k-k_0)p^r}\right).$$

Expanding the product $(S+a)^{jp^rq}(S+a)^{kp^r}$ , we obtain
 $$(S+a)^{jp^rq+kp^r} =  \sum\limits_{j_0=0}^j \sum\limits_{k_0=0}^k \binom{j}{j_0}\binom{k}{k_0}S^{j_0p^rq+k_0p^r}a^{(j-j_0)p^rq (k-k_0)p^r}.$$ To determine  $$\deg\left( (S+a)^{jp^rq+kp^r}\mod S^{q+1}-\gamma \right)$$ we need to understand $$\deg \left( S^{j_0p^rq+k_0p^r} \mod S^{q+1}-\gamma\right).$$ The bounds on $j$ and $k$ imply that $j_0+k_0 < 2\frac{q}{p^r}$.

 Case 1: $k_0 \geq j_0$. In this case let $k_0 = j_0 +\delta$ where $0\leq  \delta < \frac{q}{p^r}$. Then $$S^{j_0p^rq+k_0p^r} =  S^{j_0p^r(q+1)+\delta p^r} \equiv S^{\delta p^r}\gamma^{j_0p^r} \mod S^{q+1}-\gamma.$$
 
  Case 2: $k_0 < j_0$. In this case let $k_0 = j_0 - \delta$ where $0 <  \delta < \frac{q}{p^r}$. We may write $j_0p^r(q+1) + k_0 p^r = (j_0p^r-1)(q+1)+(q+1-(j_0-k_0) p^r)$. The bounds  on $j,k,j_0, k_0$ imply that $0 \leq q+1-(j_0-k_0)p^r < q+1$. Therefore $$S^{j_0p^rq+k_0p^r} =  S^{ (j_0p^r-1)(q+1)+(q+1-(j_0-k_0) p^r)}$$ and  $$S^{j_0p^rq+k_0p^r}  \equiv S^{q+1-(j_0-k_0)p^r}\gamma^{j_0p^r-1}  \mod S^{q+1}-\gamma.$$
 
 The reductions in case 1 give powers of the form $S^{lp^r}$ where $0 \leq l < \frac{q}{p^r}$. The reductions in case 2 give powers of the form $S^{q+1-l'p^r}$ where $1 \leq l' < \frac{q}{p^r}$. The largest possible power in case 1 is $q-p^r$. The largest possible power in case 2 is $q+1-p^r$.  Therefore $\deg \left( (S+a)^{jp^rq+kp^r} \mod S^{q+1}-\gamma\right)  \leq q+1-p^r$. Since  $$\deg \left( (S+a)^{jp^rq+kp^r} \mod S^{q+1}-\gamma\right) = \deg\left(T^{jp^rq+kp^r} \mod P_{a,\gamma}\right),$$     the result follows. \end{proof} 
 
In Lemma \ref{lem:reduction Xqpr2} the bound $kp^r < q$ is key. If $kp^r = q$,then $jp^rq +kp^r = jp^rq+q = (jp^r+1)q$. Since the highest power of $p$ dividing $jp^r+1$ is $1$, the reduction of $T^{jp^rq+kp^r} \mod P_{a,\gamma}$ has degree $q+1-1 = q$. For future reference we present a slightly more general version of Lemma  \ref{lem:reduction Xqpr2}.

\begin{corollary}\label{cor:reduction Xi}
Let $q$ be a power of $p$. Let $j$ be a positive integer relatively prime to $p$ and  let $p^r$ be a power of $p$ such that  $jp^r<q$. Let $k$ be a positive integer such that $kp^r <q$  and let $0 \leq k_1 < p^r.$ Then
$$\deg\left({T^{jp^rq+kp^r+k_1}\mod P_{a,\gamma}}\right) \leq q+1-p^r+k_1$$

\end{corollary}

The good powers of $T$, those which satisfy $\deg\left({T^{jp^rq+kp^r+k_1}\mod P_{a,\gamma}}\right) \leq q-1$, are precisely the good monomials of the form $X^a$ in both \cite[Theorem 12]{LMMPW21} when $p = 2$) and \cite[Theorem 5]{N22}  in the general case. In the next section we study the reduction of $T^j(aT+b)^k \mod P_{a, \gamma}$ to find more good mononials and improve the rate bounds of Hermitian Lifted codes.

\section{Finding good monomials}Using univariate polynomial division we have proven that $\deg \left( T^i \mod P_{a,\gamma}\right) \leq q+1-p^r$ for some $i$. We shall use this result to count monomials of the form $X^iY^j$ where $0 \leq i < q^2, 0 \leq j < q$ whose restriction to $L_{a,b}$ has low degree.
%
%



\begin{lemma}\label{lem:goodmonomialsnb}
Let $q$  be a prime power of $p$. Let $0 \leq  i< q^2$. Suppose $i = i_1q+i_2 p^r+i_3$ where $0 \leq i_1 \leq q$, $p^r$ is the highest power of $p$ dividing $i_1$,$0 \leq i_2 < \frac{q}{p^{r}}$ and $0 \leq i_3 < p^r$.  Let $j = j_2 p^r + j_3$ where $0 \leq j < q$, $0 \leq j_2 < \frac{q}{p^{r}}$ and $j_3 < p^r$. The monomial $X^{i}Y^{j}$ is a good monomial if $i_2 + j_2 < \frac{q}{p^r}$ and $ i_3+  j_3\leq p^r-2 $.\end{lemma}

\begin{proof}

We shall evaluate the monomial $X^iY^j$ on $L_{a,b}$ where $a^{q+1}+b^q+b \neq 0$. Denote $a^{q+1}+b^q+b$ by $\gamma$. In this case, we set $X = T$ and $Y = aT+b$. 
Recall that $X^iY^j$ is good if and only if $T^i(aT+b)^j \mod P_{a^q, \gamma}$ has degree $\leq q-1$. From the conditions of the theorem, $i_1 = l_1 p^r$ where $l_1$ is coprime to $p$. We need to determine the degree of $$T^{i}(aT+b)^j \mod (T-\alpha^q)^{q+1}-\gamma $$

Note that  $$T^{i}(aT+b)^j  = T^{l_1 p^r q + i_2 p^r + i_3}(aT+b)^{j_2 p^r} (aT+b)^{j_3}.$$

Therefore we rewrite the product as:

$$T^{i}(aT+b)^j  = T^{l_1 p^r q+i_2 p^r+ i_3}(a^{p^r}T^{p^r}+b^{p^r})^{j_2} (aT+b)^{j_3}.$$

Note that  $$(a^{p^r}T^{p^r}+b^{p^r})^{j_2} = \sum\limits_{u_2=0}^{j_2}\binom{j_2}{u_2} a^{u_2 p^r}b^{p^r(j_2-u_2)} T^{u_2 p^r}$$

and 

$$(aT+b)^{j_3} = \sum\limits_{u_3=0}^{j_3}\binom{j_3}{u_3} a^{u_2}b^{j_3-u_3} T^{u_3}.$$

Setting $c_{u_2,j_2,u_3,j_3} = a^{u_2 p^r}b^{p^r(j_2-u_2)}$ to simpify notation we obtain

$$T^{i}(aT+b)^j  = T^{l_1 p^r q + i_2 p^r + i_3}\left( \sum\limits_{u_2=0}^{j_2}\sum\limits_{u_3=0}^{j_3}\binom{j_2}{u_2} \binom{j_3}{u_3}c_{u_2,j_2,u_3,j_3} T^{u_2 p^r+u_3} \right)  .$$

We include the power $T^{l_1p^rq + i_2 p^r}$ into the sum 
$$T^{i}(aT+b)^j  = \left( \sum\limits_{u_2=0}^{j_2}\sum\limits_{u_3=0}^{j_3}\binom{j_2}{u_2} \binom{j_3}{u_3}c_{u_2,j_2,u_3,j_3} T^{l_1 p^r q + (i_2+u_2) p^r + i_3+u_3} \right) .$$

Our aim is to prove that the conditions on $i_1$, $i_2$, $i_3$, $j_2$ and $j_3$ imply $\deg \left( T^{l_1 p^r q + (i_2+u_2) p^r + i_3+u_3} \mod P_{a^q,\gamma} \right) < q$. Since $0 \leq u_2 + i_2 \leq j_2+i_2 < \frac{q}{p^r}$  Lemma \ref{lem:reduction Xqpr2} implies the terms of $T^{l_1p^rq+(u_2+i_2) p^r} \mod P_{a^q,\gamma}$ are $T^{s_1p^r}$ where $0 \leq s_1 \leq j_2+i_2-l_1$ and $T^{q+1-s_2p^r}$ where $1 \leq s_2 \leq l_1-j_2-i_2$. Denote $T^{l_1p^rq+(u_2+i_2) p^r} \mod P_{a^q,\gamma}$ by $f_{l_1,u_2,i_2}(T)$.  Note that the degrees of the terms of $T^{i_3}(aT+b)^{j_3}$ lie between $i_3$ and $i_3+j_3$. Since $0\leq i_3 + j_3 < p^r-2$, then all terms of $f_{k_1,u_2,i_2}(T)T^{i_3}(aT+b)^{j_3}$ are good.\end{proof}

\begin{corollary}\label{cor:mons1}

Let $p$ be a prime. Let $q = p^k$. There are at least $$\binom{q+1}{2} + \sum\limits_{r=1}^{k-1} \frac{q}{p^{r+1}}(p-1)\binom{\frac{q}{p^r}+1}{2}\binom{p^r}{2} $$ good monomials in $\mathcal{M}$.

\end{corollary}
\begin{proof}
Lemma \ref{lem:goodmonomialsnb} implies that monomials $X^iY^j$ where  $i = i_1q+i_2 p^r+i_3$,$j = j_2 p^r + j_3$   where $0 \leq i_1 \leq q$, $p^r$ is the highest power of $p$ dividing $i_1$ and $0 \leq i_2+j_2 < \frac{q}{p^r}$, $0 \leq i_3+j_3 < p^r-2$ are good.

If $i_1 = 0$, then there are $\binom{q+1}{2}$ of degree $q-1$ or less. Let $q = p^k$. Given $1 \leq r \leq k-1$ there are exactly $\frac{q}{p^{r+1}}(p-1)$ integers in $\{1,2, \ldots, q-1\}$ whose highest power of $p$ dividing them is precisely $p^r$. For each $i_1$ we count the number of possible $i_2 p^r+i_3$ and $j_2 p^r + j_3$ satisfying the conditions of the corollary. Since $0 \leq i_2+j_2 \leq \frac{q}{p^r}-1$ it follows that  there are $\binom{\frac{q}{p^r}+1}{2}$ possible pairs. Since $0 \leq i_3+j_3 \leq p^r-2$ there are $\binom{p^r}{2}$ possible pairs of $i_3$ and $j_3$.

There are a total of $$\binom{q+1}{2} + \sum\limits_{r=1}^{k-1} \frac{q}{p^{r+1}}(p-1)\binom{\frac{q}{p^r}+1}{2}\binom{p^r}{2} $$ monomials.

\end{proof}
We end this section with a corollary on the rate of the Hermitian Lifted code, $\mathcal{C}$.

\begin{corollary}
 The rate of $\mathcal{C}_f$ is at least $\frac{1}{4(p+1)}$.

\end{corollary}
\begin{proof}
Corollary \ref{cor:mons1} implies there are at least $\binom{q+1}{2} + \sum\limits_{r=1}^{k-1} \frac{q}{p^{r+1}}(p-1)\binom{\frac{q}{p^r}+1}{2}\binom{p^r}{2}$ monomials in $\mathcal{M}$ which give good functions. This implies $$\dim \mathcal{C} \geq \dim \mathcal{C}_M \geq \binom{q+1}{2} + \sum\limits_{r=1}^{k-1} \frac{q}{p^{r+1}}(p-1)\binom{\frac{q}{p^r}+1}{2}\binom{p^r}{2}.$$ Now we estimate $\lim\limits_{k \rightarrow \infty} \frac{\dim \mathcal{C}_M}{q^3}$. First we rewrite the sum in a form more amenable to limit computations.

$$\binom{q+1}{2} + \sum\limits_{r=1}^{k-1} \frac{q}{p^{r+1}}(p-1)\binom{\frac{q}{p^r}+1}{2}\binom{p^r}{2} $$
$$ = \binom{q+1}{2} + \frac{q(p-1)}{4p}\sum\limits_{r=1}^{k-1} \frac{1}{p^{r}}\left(\frac{q}{p^r}+1\right)\left(\frac{q}{p^r}\right)(p^r)(p^r-1) $$
$$ = \binom{q+1}{2} + \frac{q(p-1)}{4p}\sum\limits_{r=1}^{k-1} \left(\frac{q^2}{p^{2r}}+\frac{q}{p^r}\right)(p^r-1) $$
$$ = \binom{q+1}{2} + \frac{q(p-1)}{4p}\sum\limits_{r=1}^{k-1} \left(\frac{q^2}{p^{r}}+q -\frac{q^2}{p^{2r}}-\frac{q}{p^r}\right) $$
$$ = \binom{q+1}{2} + \frac{q(p-1)}{4p}\sum\limits_{r=1}^{k-1} \left(\frac{q^2-q}{p^{r}}+q -\frac{q^2}{p^{2r}}\right).$$

We take the sum of the corresponding geometric series and obtain
$$ \binom{q+1}{2} + \frac{q(p-1)}{4p}\left(\frac{q^2-q}{p}\left(\frac{1-\frac{1}{p^{k-1}}}{1-\frac{1}{p}} \right) -\frac{q^2}{p^2}\left(\frac{1-\frac{1}{p^{2(k-1)}}}{1-\frac{1}{p^2}} \right)     +q   \right) $$

Dividing that expression by $q^3$ we obtain 
$$  \frac{1}{q^3}\binom{q+1}{2} + \frac{(p-1)}{4p}\left(\frac{1-\frac{1}{q}}{p}\left(\frac{1-\frac{1}{p^{k-1}}}{1-\frac{1}{p}} \right) -\frac{1}{p^2}\left(\frac{1-\frac{1}{p^{2(k-1)}}}{1-\frac{1}{p^2}} \right)     +\frac{1}{q^2}  \right) $$
Since this sum of geometric series is monotone decreasing $ \frac{\dim \mathcal{C}_M}{q^3}$ is bounded below by $$ \lim\limits_{k \rightarrow \infty} \frac{(p-1)}{4p}\left(\frac{1-\frac{1}{q}}{p}\left(\frac{1-\frac{1}{p^{k-1}}}{1-\frac{1}{p}} \right) -\frac{1}{p^2}\left(\frac{1-\frac{1}{p^{2(k-1)}}}{1-\frac{1}{p^2}} \right)     +\frac{1}{q^2}  \right)$$
$$ = \frac{(p-1)}{4p}\left(\frac{1}{p}\left(\frac{1}{1-\frac{1}{p}} \right) -\frac{1}{p^2}\left(\frac{1}{1-\frac{1}{p^2}} \right)    \right) = \frac{(p-1)}{4p}\left(\frac{1}{p-1} -\frac{1}{p^2-1}\right)$$
$$ = \frac{1}{4p}\left(1 -\frac{1}{p+1}\right) = \frac{1}{4p}\left(\frac{p}{p+1}\right) = \frac{1}{4(p+1)}$$
\end{proof}
\subsection{ Even characteristic case}

If $q$  is a power of $2$ certain binomial coefficients are zero. This implies we can find more good monomials from the expansion of $(aT+b)^j$ and improve the dimension bound further. For this subsection suppose that $q = 2^k$. We shall use the following result from \cite{LMMPW21}.

\begin{proposition}\cite[Theorem 10]{LMMPW21}\label{prop:goodred}
Let $i = i_1q+i_2 2^r+ i_3$ where $0 \leq i_1, i_2, i_3$ where $2^r$ is the highest power of $2$ dividing $i_1$,$0 \leq i_2 < \frac{q}{2^{r}}$ and $0 \leq i_3 < 2^r-1$. Then $\deg\left( T^i \mod P_{a, \gamma}\right) < q$.
\end{proposition}

\begin{lemma}\label{lem:goodmonomialsnb}

Let $i = i_1q+i_2 2^r+ i_3$ where $0 \leq i_1, i_2, i_3$ where $2^r$ is the highest power of $2$ dividing $i_1$,$0 \leq i_2 < \frac{q}{2^{r}}$ and $0 \leq i_3 < 2^r$ Let $j = j_2 2^r + j_3$ where $0 \leq j < q$, $0 \leq j_2 < \frac{q}{2^{r}}$ and $j_3 < 2^r$. The monomial $X^{i}Y^{j}$ is a good monomial if $i_2 2^r + i_3 +j_2 2^r + j_3 < q$ and $ 2^r-1-i_3 \not \leq_2 j_3 $
\end{lemma}
\begin{proof}
We proceed as in the proof of Lemma \ref{lem:goodmonomialsnb} and obtain that $$T^{i}(aT+b)^j  = \left( \sum\limits_{u_2=0}^{j_2}\sum\limits_{u_3=0}^{j_3}\binom{j_2}{u_2} \binom{j_3}{u_3}c_{u_2,j_2,u_3,j_3} T^{l_1 2^r q + (i_2+u_2) 2^r + i_3+u_3} \right) .$$

Our aim is to prove that the conditions on $i_1$, $i_2$, $i_3$, $j_2$ and $j_3$ imply $T^{i}(aT+b)^j  \mod P_{a^q, \gamma}$ is good. Proposition \ref{prop:goodred} implies that if $ (i_2+u_2) 2^r + i_3+u_3 < q$ and $i_3 + u_3 \not\equiv -1 \mod 2^r$ then $\deg \left( T^{l_1 2^r q + (i_2+u_2) 2^r + i_3+u_3} \right) < q$. If $u_3 + i_3 = 2^r-1$ it implies $u_3 = 2^r-1-i_3$. Since $u_3 = 2^r-1-i_3 \not \leq_2 j_3$ it follows that $\binom{j_3}{u_3} = \binom{j_3}{2^r-1-i_3} = 0$.

Therefore given the conditions on $i_1$, $i_2$, $i_3$, $j_2$ and $j_3$ it follows that either $\deg \left( T^{l_1 2^r q + (i_2+u_2) 2^r + i_3+u_3} \right) < q$ or  $\binom{j_3}{u_3} = \binom{j_3}{2^r-1-i_3} = 0$. This implies  $$\deg\left(T^{i}(aT+b)^j \mod P_{a, \gamma}\right) < q .$$ \end{proof}
\begin{corollary}\label{cor:mons2}

Let $q = 2^k$. There are at least $$\binom{q+1}{2} + \sum\limits_{r=1}^{k-1} \frac{q}{2^{r+1}}\left(\binom{\frac{q}{2^r}}{2}(4^r-3^r) + \frac{q}{2^r}\binom{2^r}{2}\right) $$ good monomials in $\mathcal{M}$.

\end{corollary}
\begin{proof}
Lemma \ref{lem:goodmonomialsnb} implies that monomials $X^iY^j$ where  $i = i_1q+i_2 2^r+ i_3$, $j_2q^r+j_3 $ where $0 \leq i_1, i_2, i_3$ where $2^r$ is the highest power of $2$ dividing $i_1$,$0 \leq i_2 < \frac{q}{2^{r}}$ and $0 \leq i_3 < 2^r$,  $0 \leq j < q$, $0 \leq j_2 < \frac{q}{2^{r}}$ , $j_3 < 2^r$ and  $i_2 2^r + i_3 +j_2 2^r + j_3 < q$ and $ 2^r-1-i_3 \not \leq_2 j_3 $ are good.

If $i_1 = 0$, then there are $\binom{q+1}{2}$ monomials of degree $q-1$ or less. Let $q = 2^k$. Given $1 \leq r \leq k-1$ there are exactly $\frac{q}{2^{r+1}}$ integers in $\{1,2, \ldots, q-1\}$ whose highest power of $2$ dividing them is precisely $2^r$. For each $i_1$ we count the number of possible $i_2 2^r+i_3$ and $j_2 2^r + j_3$ satisfying the conditions of the corollary. 

If $0 \leq i_2+ j_2 \leq \frac{q}{2^r}-2$, all possible values of $i_3$ and $j_3$ satisfy $i_2 2^r + i_3 +j_2 2^r + j_3 < q$. The only values which do not satisfy the conitions of the theorem are $i_3 = 2^r-1-j_3$ where $i_3 \leq_2 j_3$. There are $\binom{\frac{q}{2^r}}{2}$ values for $i_2$ and $j_2$ and $4^r-3^r$ values for $i_3$ and $j_3$. 

If $i_2 + j_2 = \frac{q}{2^r}-1$ then all values of $i_3$, $j_3$ such that $i_3+j_3 \leq 2^r-2$ satisfy the conditions of the corollary. There are $\frac{q}{2^r}$ values for $i_2$ and $j_2$ and $\binom{2^r}{2}$ values for $i_3$ and $j_3$. In total there are at least

 $$\binom{q+1}{2} + \sum\limits_{r=1}^{k-1} \frac{q}{2^{r+1}}\left(\binom{\frac{q}{2^r}}{2}(4^r-3^r) + \frac{q}{2^r}\binom{2^r}{2}\right) $$ good monomials \end{proof}

We end this subsection with an improvement to corollary on the rate of the Hermitian Lifted code, $\mathcal{C}$.

\begin{corollary}
 Let $q$ be a power of $2$. The rate of $\mathcal{C}$ is at least $\frac{1}{10}$.

\end{corollary}
\begin{proof}

As in the previous case, we find the limit of the number of good monomials divided by $q^3$.  Corollary \ref{cor:mons2} implies  there are at least $$\binom{q+1}{2} + \sum\limits_{r=1}^{k-1} \frac{q}{2^{r+1}}\left(\binom{\frac{q}{2^r}}{2}(4^r-3^r) + \frac{q}{2^r}\binom{2^r}{2}\right)$$ monomials which give good functions. This implies $$\dim \mathcal{C} \geq \dim \mathcal{C}_M \geq \binom{q+1}{2} + \sum\limits_{r=1}^{k-1} \frac{q}{2^{r+1}}\left(\binom{\frac{q}{2^r}}{2}(4^r-3^r) + \frac{q}{2^r}\binom{2^r}{2}\right).$$ Now we estimate $\lim\limits_{k \rightarrow \infty} \frac{\dim \mathcal{C}}{q^3}$. First we rewrite the sum in a form more amenable to limit computations.

Note that 
 $$\dim \mathcal{C} \geq \dim \mathcal{C}_M\geq \binom{q+1}{2} + \sum\limits_{r=1}^{k-1} \frac{q}{2^{r+1}}\left(\binom{\frac{q}{2^r}}{2}(4^r-3^r) + \frac{q}{2^r}\binom{2^r}{2}\right) $$
  $$= \frac{q^2+q}{2} +\frac{q}{4} \sum\limits_{r=1}^{k-1} \frac{1}{2^{r}}\left( \frac{q}{2^r}\left(\frac{q}{2^r}-1\right)\left(4^r-3^r\right ) + \frac{q}{2^r}\left(2^r\right)\left(2^r-1\right)\right) $$
    $$= \frac{q^2+q}{2} +\frac{q}{4} \sum\limits_{r=1}^{k-1} \frac{1}{2^{r}}\left( \frac{q}{2^r}\left((2)^rq-\left(\frac{3}{2}\right)^rq-4^r+3^r \right)+ q\left(2^r-1\right)\right) $$
        $$= \frac{q^2+q}{2} +\frac{q}{4} \sum\limits_{r=1}^{k-1} \frac{1}{2^{r}}\left( \left(q^2-\left(\frac{3}{4}\right)^rq^2-2^rq+\left(\frac{3}{2}\right)^rq \right)+ 2^rq-q\right) $$
        $$= \frac{q^2+q}{2} +\frac{q}{4} \sum\limits_{r=1}^{k-1} \frac{1}{2^{r}}\left( q^2-\left(\frac{3}{4}\right)^rq^2+\left(\frac{3}{2}\right)^rq -q\right) $$
                $$= \frac{q^2+q}{2} +\frac{q}{4} \sum\limits_{r=1}^{k-1} \left( \left(\frac{1}{2}\right)^rq^2-\left(\frac{3}{8}\right)^rq^2+\left(\frac{3}{4}\right)^rq -\left(\frac{1}{2}\right)^rq\right).$$
                
We sum the geometric terms and obtain:$$= \frac{q^2+q}{2} +\frac{q}{4} \sum\limits_{r=1}^{k-1} \left( \left(1-\frac{1}{2^{k-1}}\right)\left(q^2-q\right)-\left(1-\frac{3^{k-1}}{8^{k-1}}\right)\left(\frac{3}{5}\right)q^2+ \left(1-\frac{3^{k-1}}{4^{k-1}}\right)3q\right).$$

If we divide by $q^3$ and take the limit as $k \rightarrow \infty$ we obtain $$R\rightarrow \frac{1}{4}\left(1-\frac{3}{5}\right) = \frac{1}{10}.$$ \end{proof}



\section{Minimum distance}

If $P$ represents a nonzero position of a codeword in the LRC code $\mathcal{C}$, it implies that, since each point of the Hermitian curve lies on $q^2-1$ lines, there must be at least $q^2-1$ other positions that must be zero. For each of those additional $q^2-1$ nonzero positions, they must contain another $q^2-2$ nonzero positions As a result, our code has a minimum distance of at least $q^2$. However, the code $\mathcal{C}_M$  is constructed by evaluating good monomials from $\mathcal{M}$. The weighted degree of $X^iY^j$ is given by $qi+(q+1)j$. For a Hermitian one-point code, if we have $i \leq q^2-q-1$ and $f(X,Y) = \sum f_{i,j}X^iY^j$, where $qi+(q+1)j \leq s$, it will have at least $q^3-s$ nonzero values when evaluated on the Hermitian curve.

The monomial with the highest weighted degree that qualifies as a good function is $i_1 = q-p$, $i_2,i_3 = 0$, $j_2 = \frac{q}{p}-1$, and $j_3 = p-2$. This means that $i = (q-p)q$ and $j = (\frac{q}{p}-1)p+(p-2) = q-2$. The minimum distance of the corresponding code is $q^3 - q(q^2-pq)-(q+1)(q-2) = pq^2 -(q^2-q-2) = (p-1)q^2 -q+2$. 

Computational analysis has shown that there are non-monomial functions that also qualify as good functions. For example for $q = 2$, $\dim \mathcal{C} = \mathcal{C}_M = 3$. But for  $q = 4$, $\dim \mathcal{C} = 16$ and $ \mathcal{C}_M = 13$ and for $q = 8$   $\dim \mathcal{C} = 75$ and $ \mathcal{C}_M = 111$. Identifying the nonmonomial good functions would grealy enhance the dimension bound of Hermitian Lifted codes. It is unclear if the minimum distance would be greatly reduced or not.

\section{Comparison with other codes}

One of the reasons to build LRCs from Hermitian codes is to compare them with Reed--Solomon lifted codes. A Reed--Solomon lifted code of length $N = q^2$ has locality $q-1$ and availability $s = q+1 = \sqrt{N} +1$. Reed--Solomon codes lifted codes have dimension $q^2-3^r$ where $q=2^r$.  Hermitian lifted codes have much larger availability. Their length is $N = q^3$, their locality is $q+1$ and their availability is $s = q^2-1 = \sqrt[3]{N^2}-1$. The dimension of Hermitian Lifted codes is much smaller, but this is to be expected since their availability is much greater. The information rate of lifted Reed--Solomon codes tends to $1$ whereas the information rate of Hermitian lifted codes tends to $0.1$. It can be difficult to compare both codes, but we hope this construction can be extended to other algebraic and projective varieties.

\section*{Conclusion}

We have enhanced the dimension bound of a particular class of Locally Recoverable codes derived from the Hermitian curve by employing functions that exhibit low-degree polynomial behavior on each line of the curve. The rate of Hermitian-Lifted codes is significantly lower than that of Reed-Solomon lifted codes. However, it is noteworthy that we can construct codes with positive rates using parity check equations over a field $\mathbb{F}_q$, even though the corresponding binary vector code would have zero dimension. We anticipate that this advancement will pave the way for codes with improved rates derived from the Hermitian unital and other similar designs

\end{document}